\documentclass[smallextended,final,twoside]{svjour3}

\usepackage{graphicx,epic,eepic,epsfig,amsmath,latexsym,amssymb,verbatim,color,revsymb}
\usepackage{theorem}

\newtheorem{metalemma}[definition]{Lemma'}

\def\squareforqed{\hbox{\rlap{$\sqcap$}$\sqcup$}}
\def\qed{\ifmmode\squareforqed\else{\unskip\nobreak\hfil
\penalty50\hskip1em\null\nobreak\hfil\squareforqed
\parfillskip=0pt\finalhyphendemerits=0\endgraf}\fi}
\def\endenv{\ifmmode\;\else{\unskip\nobreak\hfil
\penalty50\hskip1em\null\nobreak\hfil\;
\parfillskip=0pt\finalhyphendemerits=0\endgraf}\fi}

\mathchardef\ordinarycolon\mathcode`\:
\mathcode`\:=\string"8000
\def\vcentcolon{\mathrel{\mathop\ordinarycolon}}
\begingroup \catcode`\:=\active
  \lowercase{\endgroup
  \let :\vcentcolon
  }

\newcommand{\nc}{\newcommand}
\nc{\rnc}{\renewcommand}
\nc{\beq}{\begin{equation}}
\nc{\eeq}{{\end{equation}}}
\nc{\beqa}{\begin{eqnarray}}
\nc{\eeqa}{\end{eqnarray}}
\nc{\lbar}[1]{\overline{#1}}
\nc{\bra}[1]{\langle#1|}
\nc{\ket}[1]{|#1\rangle}
\nc{\ketbra}[2]{|#1\rangle\!\langle#2|}
\nc{\braket}[2]{\langle#1|#2\rangle}
\nc{\proj}[1]{| #1\rangle\!\langle #1 |}
\nc{\avg}[1]{\langle#1\rangle}
\nc{\Rank}{\operatorname{Rank}}
\nc{\smfrac}[2]{\mbox{$\frac{#1}{#2}$}}
\nc{\tr}{\operatorname{Tr}}
\nc{\ox}{\otimes}
\nc{\dg}{\dagger}
\nc{\dn}{\downarrow}
\nc{\cA}{{\cal A}}
\nc{\cB}{{\cal B}}
\nc{\cC}{{\cal C}}
\nc{\cD}{{\cal D}}
\nc{\cE}{{\cal E}}
\nc{\cF}{{\cal F}}
\nc{\cG}{{\cal G}}
\nc{\cH}{{\cal H}}
\nc{\cI}{{\cal I}}
\nc{\cJ}{{\cal J}}
\nc{\cK}{{\cal K}}
\nc{\cL}{{\cal L}}
\nc{\cM}{{\cal M}}
\nc{\cN}{{\cal N}}
\nc{\cO}{{\cal O}}
\nc{\cP}{{\cal P}}
\nc{\cR}{{\cal R}}
\nc{\cS}{{\cal S}}
\nc{\cT}{{\cal T}}
\nc{\cX}{{\cal X}}
\nc{\cZ}{{\cal Z}}
\nc{\csupp}{{\operatorname{csupp}}}
\nc{\qsupp}{{\operatorname{qsupp}}}
\nc{\var}{{\operatorname{var}}}
\nc{\rar}{\rightarrow}
\nc{\lrar}{\longrightarrow}
\nc{\polylog}{{\operatorname{polylog}}}
\nc{\1}{{\openone}}
\nc{\wt}{{\operatorname{wt}}}

\nc{\RR}{{{\mathbb R}}}
\nc{\CC}{{{\mathbb C}}}
\nc{\FF}{{{\mathbb F}}}
\nc{\NN}{{{\mathbb N}}}
\nc{\ZZ}{{{\mathbb Z}}}
\nc{\PP}{{{\mathbb P}}}
\nc{\QQ}{{{\mathbb Q}}}
\nc{\UU}{{{\mathbb U}}}
\nc{\EE}{{{\mathbb E}}}
\nc{\id}{{\operatorname{id}}}

\nc{\CHSH}{{\operatorname{CHSH}}}

\nc{\be}{\begin{equation}}
\nc{\ee}{{\end{equation}}}
\nc{\bea}{\begin{eqnarray}}
\nc{\eea}{\end{eqnarray}}
\nc{\<}{\langle}
\rnc{\>}{\rangle}
\nc{\Hom}[2]{\mbox{Hom}(\CC^{#1},\CC^{#2})}
\nc{\rU}{\mbox{U}}

\nc{\ob}[1]{#1}

\begin{document}

\title{Squashed entanglement, $\mathbf{k}$-extendibility, \protect\\ quantum Markov chains, and recovery maps}
\titlerunning{Squashed entanglement and recovery maps}

\author{Ke Li${}^1$ \and Andreas Winter${}^2$}
\authorrunning{Ke Li \and Andreas Winter}

\institute{${}^1$KL is with Institute for Advanced Study in Mathematics, Harbin Institute of Technology, Harbin 150006, P. R. China. \\\email{carl.ke.lee@gmail.com} \\
${}^2$AW is with ICREA \& Departament de F\'isica: Grup d'Informaci\'o Qu\`antica, Universitat Aut\`onoma de Barcelona, 01893 Bellaterra (Barcelona), Spain. \\\email{andreas.winter@uab.cat}}

\date{24 December 2017}

\maketitle

\begin{abstract}
Squashed entanglement [Christandl and Winter, \emph{J. Math. Phys.} 45(3):829-840 (2004)]
is a monogamous
entanglement measure, which implies that highly extendible states have small
value of the squashed entanglement. Here, invoking a recent inequality for the
quantum conditional mutual information [Fawzi and Renner,
\emph{Commun. Math. Phys.} 340(2):575-611 (2015)]
greatly extended and simplified in various work since, 
we show the converse, that a small value of squashed entanglement implies
that the state is close to a highly extendible state. As a corollary, we
establish an alternative proof of the faithfulness of squashed entanglement
[Brand\~{a}o, Christandl and Yard, \emph{Commun. Math. Phys.} 306:805-830 (2011)].

We briefly discuss the previous and subsequent 
history of the Fawzi-Renner bound and related conjectures, 
and close by advertising a potentially far-reaching generalization
to universal and functorial recovery maps for the monotonicity 
of the relative entropy.
\end{abstract}

\bigskip\noindent
{\bf Squashed entanglement.}---%
One of the core goals in the theory of entanglement is its quantification,
for which purpose a large number of either operationally or mathematically/axiomatically
motivated entanglement measures and monotones have been introduced and studied
intensely since the 1990s~\cite{H4:review,Christ:PhD}.

In this paper we will discuss one specific such measure, the so-called
\emph{squashed entanglement}~\cite{ChristandlWinter}, defined as
\begin{equation}
  E_{\text{sq}}(\rho^{AB})
     := \inf \frac12 I(A:B|E) \text{ s.t.} \tr_E \rho^{ABE}=\rho^{AB},
\end{equation}
where $I(A:B|E) = S(AE)+S(BE)-S(E)-S(ABE)$ is the \emph{(quantum) conditional
mutual information}, which by strong subadditivity of the von Neumann entropy
is always non-negative~\cite{SSA}; and $\rho^{ABE}$ as above is called
an \emph{extension} of $\rho^{AB}$.
This definition appears to have been put forward first in~\cite{Tucci},
where it was also remarked that by restricting the extension of $\rho^{AB}$
to have the form
$\rho^{ABE} = \sum_i p_i \proj{\varphi_i}^{AB} \ox \proj{i}^E$,
the minimization reduces to the well-known
\emph{entanglement of formation}~\cite{BDSW1996},
\begin{equation}
  E_F(\rho^{AB}) = \min \sum_i p_i S(\varphi_i^A) \text{ s.t.} \sum_i p_i \proj{\varphi_i} = \rho.
\end{equation}

While it is fairly straightforward to see from their definitions that both
$E_{\text{sq}}$ and $E_F$ are convex functions of the state, the former has
many properties that the latter lacks, among them additivity and
monogamy~\cite{ChristandlWinter,KoashiWinter} as well 
as~\cite{aspects,monogamology}, cf.~\cite{monogamous,Christ:PhD}.
Namely, abbreviating $E_{\text{sq}}(\rho^{AB}) = E_{\text{sq}}(A:B)$,
we have
\begin{equation}
  \label{eq:sq-monogam-gen}
  E_{\text{sq}}(A:B_1B_2) \geq E_{\text{sq}}(A:B_1) + E_{\text{sq}}(A:B_2).
\end{equation}
In particular, if $\rho^{AB}$ is $k$-extendible, meaning that there exists
a state $\rho^{AB_1\ldots B_k}$ such that $\rho^{AB} = \rho^{AB_i}$ for all
$i$ (and that w.l.o.g.~is symmetric with respect to permutations of the $B$-systems), then
\begin{equation}
  \label{eq:sq-monogam}
  E_{\text{sq}}(A:B) \leq \frac{1}{k}\log|A|.
\end{equation}
While clearly $E_{\text{sq}} \leq E_F$, in the other direction,
squashed entanglement is an upper bound on the distillable entanglement
and indeed on the distillable secret key in a state~\cite{ChristandlWinter,Christ:PhD},
which makes it very useful to the theory of state distillation and
channel capacities, cf.~\cite{Wilde:E-sq}.

One of the properties much desirable for a quantitative entanglement
measure is \emph{faithfulness}, i.e.~the fact that it is zero if and
only if the state is separable, and otherwise strictly positive. To be
truly useful, such a statement ought to come in the form of a
relationship between the value of the entanglement measure, and a
suitably chosen distance from the set of separable states. After being 
an open problem for a while, this was finally obtained a few years ago by 
Brand\~{a}o \emph{et al.}~\cite{BCY}, and later improved by us~\cite{LiWinter}.

In the present paper, we will reproduce this finding in a conceptually
simple and appealing way, by first showing a relation between the
value of squashed entanglement and the distance from $k$-extendible states,
and then invoking a suitable de Finetti theorem to bound the distance from
separable states. (That in the limit of $k\rightarrow\infty$ the state has
to be separable was known for some time~\cite{SchumacherWerner:sharable},
but we shall use more recent, quantitative, versions.)
We go on to contrast this finding with the faithfulness
of entanglement of formation. Then, we put the technical result of
Fawzi and Renner~\cite[Thm.~5.1]{FawziRenner}, on which our proof
crucially relies, in the context of other conjectured inequalities and subsequent
results; motivated by a much more general observation in classical
probability, we propose as an open problem to find the ``right'' quantum
generalization.


\bigskip\noindent
{\bf Main result.}---%
Now we show that the monogamy bound, Eq.~(\ref{eq:sq-monogam}), 
has a partial converse:
\begin{theorem}
  \label{thm:main}
  Consider a state $\rho^{AB}$ with $E_{\text{sq}}(\rho) \leq \epsilon$.
  Then, for every integer $k$, there exists a $k$-extendible state
  $\sigma^{AB}$ such that $\|\rho-\sigma\|_1 \leq (k-1)\sqrt{2\ln 2}\sqrt{\epsilon}$.
  In particular, $\rho$ is $O\left(\sqrt[4]{\epsilon}\right)$-close to
  a $\Omega\left(\frac{1}{\sqrt[4]{\epsilon}}\right)$-extendible state.
\end{theorem}

\begin{corollary}
  \label{cor:sq-faithful}
  For every state $\rho^{AB}$ with $E_{\text{sq}}(\rho) \leq \epsilon$, 
  there exists a separable state $\sigma$ with
  \[
    \|\rho-\sigma\|_1 \leq 3.1 |B| \sqrt[4]{\epsilon}.
  \]
  In particular, squashed entanglement is \emph{faithful}: $E_{\text{sq}}(\rho)=0$
  if and only if the state $\rho$ is separable.
\end{corollary}

\medskip
For comparison, the earlier result of Brand\~{a}o \emph{et al.}~\cite[Cor.~1]{BCY}
yields
\begin{equation}
  \label{eq:BCY}
  \|\rho-\sigma\|_1 \leq \sqrt{|A| |B|}\|\rho-\sigma\|_2 \leq 12 \sqrt{|A| |B|} \sqrt{\epsilon}.
\end{equation}
The Hilbert-Schmidt ($2$-)norm bound seems not available with our techniques,
but the trace ($1$-)norm behaviour is qualitatively reproduced here, albeit with
a worse polynomial dependence on $\epsilon$ but with a slightly better constant.
In particular, it is perhaps of interest that in our bound in Corollary~\ref{cor:sq-faithful}
only the dimensionality of one of the two systems appears
(cf.~however~\cite[Eq.~(66)]{BF:de-Finetti}).

\medskip
The proof of this theorem relies essentially on a recent result
by Fawzi and Renner~\cite{FawziRenner}, stating that for every
tripartite state $\rho^{AEB}$ there exists a cptp map
$\widetilde{R}:\mathcal{L}(E) \rightarrow \mathcal{L}(EB)$ such that
\begin{equation}
  \label{eq:FR}
  -\log F\bigl(\rho^{AEB},(\id_A\ox\widetilde{R})\rho^{AE}\bigr)^2 \leq I(A:B|E)_\rho,
\end{equation}
with the fidelity $F$ of two states $\alpha$ and $\beta$ defined as
$F(\alpha,\beta) = \| \sqrt{\alpha}\sqrt{\beta} \|_1$.

\medskip
\begin{proof}
Choose an extension $\rho^{ABE}$ for $\rho^{AB}$, and use the
map $\widetilde{R}$ from Eq.~(\ref{eq:FR}).
Now we employ a basic inequality from~\cite[Thm.~1]{FuchsvdG}, saying
\begin{equation}
  1-F(\alpha,\beta) \leq \frac12 \|\alpha-\beta\|_1 \leq \sqrt{1-F(\alpha,\beta)^2},
\end{equation}
for the fidelity $F(\alpha,\beta)=\|\sqrt{\alpha}\sqrt{\beta}\|_1$.
Hence, from Eq.~(\ref{eq:FR}),
\[
  t := \sqrt{4\ln 2\,I(A:B|E)} \geq \| \rho^{AEB} - (\id_A\ox\widetilde{R})\rho^{AE} \|_1.
\]
But since $(\id_A\ox\widetilde{R})\rho^{AE} \approx \rho^{AEB}$, we may apply the
same map again, say $k-1$ times, always to the $E$ system of $\rho^{AEB}$,
arriving at a state
\[
  \omega^{AEB_1\ldots B_k}
     = (\id_A \ox \widetilde{R}^{E\rightarrow EB_k}\circ \cdots
                          \circ \widetilde{R}^{E\rightarrow EB_2})\rho^{AEB_1},
\]
which has the property that for each $i$,
$\| \omega^{AB_i} - \rho^{AB} \|_1 \leq (i-1)t$, by the triangle
inequality and the contractive property of the trace norm under
cptp maps.
Hence, tracing out $E$ and considering the symmetrization
of the $B$ systems, i.e.
\[
  \Omega^{AB_1\ldots B_k}
    = \frac{1}{k!}\sum_{\pi\in S_k} (\1\ox U^{\pi})\omega^{AB_1\ldots B_k}(\1\ox U^{\pi})^\dagger,
\]
we have that it is manifestly permutation symmetric on the $B$ systems,
and for all $i$,
\begin{equation}
  \| \Omega^{AB_i} - \rho^{AB} \|_1 \leq \frac{k-1}{2}t.
\end{equation}
Minimizing over all extensions as required by the definition
of squashed entanglement, allowing $I(A:B|E)$ to get arbitrarily
close to $2\epsilon$, concludes the proof of the theorem.
\qed

\medskip
To show the corollary, we use~\cite[Thm.~2 \& Cor.~5]{NOP}
or alternatively \cite[Thm.~II.7']{one-and-half-deFinetti},
which say that a $k$-extendible state is at trace distance at most
$\frac{2|B|^2}{k}$ from a separable state. To use the former result,
which requires Bose-symmetric extensions, we have to go from
the permutation symmetric $\Omega^{AB_1\ldots B_k}$ to a permutation
invariant purification
\[
  \ket{\Psi}^{AA'B_1B_1'\ldots B_kB_k'} \\
            = \left(\sqrt{\Omega^{AB_1\ldots B_k}}\ox\1\right)
                    \ket{\Phi}^{AA'}\ket{\Phi}^{B_1B_1'}\cdots\ket{\Phi}^{B_kB_k'},
\]
with the non-normalized maximally entangled state $\ket{\Phi} = \sum_i \ket{i}\ket{i}$.
The choice
\[
  k = \left\lfloor \sqrt[4]{\frac{2}{\ln 2}}\frac{|B|}{\sqrt[4]{\epsilon}} \right\rfloor
\]
then does the rest.
\qed
\end{proof}

\bigskip\noindent
{\bf Comparison with entanglement of formation.}---%
It is instructive to compare the monogamy relation Eq.~(\ref{eq:sq-monogam})
and its ``converse'', Theorem~\ref{thm:main} for the squashed entanglement,
with the analogous statements for the entanglement of formation:
\begin{proposition}
  \label{prop:EoF}
  In a bipartite system $AB$, with $d=\min\{|A|,|B|\}$, if the state $\rho^{AB}$ 
  is $\delta$-close in trace norm to a separable state $\sigma^{AB}$, 
  meaning that $\frac12 \|\rho-\sigma\|_1 \leq\delta$, then
  \begin{equation}
    E_F(\rho) \leq \sqrt{\delta} \log d
                     + (1+\sqrt{\delta}) h_2\!\left(\!\frac{\sqrt{\delta}}{1+\sqrt{\delta}}\!\right).
  \end{equation}
  Conversely, if $E_F(\rho) \leq \epsilon$, then this implies that there is a 
  separable state $\sigma$ with $\frac12 \|\rho-\sigma\|_1 \leq \sqrt{\ln 2}\sqrt{\epsilon}$.
\end{proposition}
\begin{proof}
The first part is originally due to Nielsen~\cite{Nielsen:cont},
with a slightly different form of the bound. The present almost optimal 
bound is from~\cite[Cor.~4]{Winter:S-continuity}.
For the second part, consider an optimal decomposition
$\rho = \sum_i p_i \proj{\varphi_i}$, such that
\[\begin{split}
  \epsilon  \geq \sum_i p_i \frac12 I(A:B)_{\varphi_i}
           &\geq \sum_i p_i \frac{1}{4\ln 2}\|\varphi_i^{AB} - \varphi_i^{A}\ox\varphi_i^{B}\|_1^2   \\
           &\geq \frac{1}{4\ln 2} \left\| \rho - \sum_i p_i \varphi_i^{A}\ox\varphi_i^{B} \right\|_1^2,
\end{split}\]
and the right hand state inside the trace norm is manifestly separable.
\qed
\end{proof}

\medskip
In other words, while entanglement of formation is essentially about the
distance from separable states, squashed entanglement is about the distance
from highly extendible states (up to log-dimensionality factors and polynomial
relation of $\epsilon$ and $\delta$). Note that squashed entanglement, like
the entanglement of formation, is \emph{asymptotically continuous}~\cite{H4:review}:
Alicki and Fannes~\cite{AlickiFannes} showed that for
$\frac12 \|\rho^{AB}-\sigma^{AB}\|_1 \leq \epsilon \leq 1$,
\(
  \displaystyle
  \bigl| E_{\text{sq}}(\rho)-E_{\text{sq}}(\sigma) \bigr| \leq 16\epsilon\log|A| + 4H_2(2\epsilon),
\)
where $H_2(x) = -x\log x-(1-x)\log(1-x)$ is the binary entropy.
Using the bounds presented in \cite{Winter:S-continuity}, it can be
improved to 
\[
  \bigl| E_{\text{sq}}(\rho)-E_{\text{sq}}(\sigma) \bigr| 
            \leq 4\epsilon\log|A| + 2(1+\epsilon) h_2\!\left(\!\frac{\epsilon}{1+\epsilon}\!\right).
\]

This explains the occurrence of states such as the $d\times d$
fully antisymmetric state $\alpha_d$, which is at trace distance $1$ from
the separable states for all $d$, but has
$E_{\text{sq}}(\alpha_d) \leq \frac{2}{d}$ which is arbitrarily small for
large $d$~\cite{antisymm}.
Indeed, this state is $(d-1)$-extendible, so by monogamy of $E_{\text{sq}}$
it has to have small squashed entanglement.
Conversely by Theorem~\ref{thm:main}, this is the
only way in which a state can have small squashed entanglement. On the
other hand, the large distance from separable, and the dimension-dependent
constants in Corollary~\ref{cor:sq-faithful} and Eq.~(\ref{eq:BCY}),
are entirely due to the fact that in large dimension, quite
highly extendible states can be far away from being separable.

\bigskip\noindent
{\bf Recovery maps and related facts \& conjectures.}---%
The form (\ref{eq:FR}) of the Fawzi-Renner bound~\cite{FawziRenner}
was arrived at in a succession of speculative steps.
The initial insight is no doubt
Petz's~\cite{Petz:sufficiency}, who showed a general statement on
the relative entropy
\[
  D(\rho\|\sigma) = \tr\rho(\log\rho-\log\sigma).
\]
Indeed, while for any two states $\rho$ and $\sigma$ on a system
$H$ and a cptp map $T:\mathcal{L}(H)\rightarrow\mathcal{L}(K)$,
$D(\rho\|\sigma) \geq D(T\rho\|T\sigma)$ -- this is equivalent to
strong subadditivity~\cite{SSA} --, Petz showed that equality holds
if and only if there exists a cptp map $R$ such that $RT\sigma=\sigma$
and $RT\rho=\rho$. What is more, this map can be constructed in a
unified way from $T$ and $\sigma$ alone, as the \emph{transpose channel},
or \emph{Petz recovery map} $R=R(T,\sigma)$, given by
\begin{equation}
  \label{eq:q-transpose}
  R(\xi) = \sqrt{\sigma}\, T^*\!\left( (T\sigma)^{-1/2} \xi (T\sigma)^{-1/2} \right)\! \sqrt{\sigma},
\end{equation}
where $T^*$ is the adjoint map to $T$, at least in the finite dimensional case
(cf.~\cite{BarnumKnill}). These transpose channels have found increasing 
attention in recent years, see e.g.~\cite{Woods-thermo,LamiWilde-Gaussian}.

The above problem involving the conditional mutual information
is recovered by letting $T=\tr_B$, $\rho = \rho^{AEB}$ and
$\sigma = \rho^A \ox \rho^{EB}$, where it can be checked that
\[\begin{split}
  I(A:B|E) &= I(A:EB)-I(A:E) \\
           &= D(\rho^{AEB}\|\rho^A\ox\rho^{EB}) - D(\rho^{AE}\|\rho^A\ox\rho^E).
\end{split}\]
In this case, the Petz recovery map reads
\begin{equation}
  \label{eq:Petz-I}
  R(\xi) = \sqrt{\rho^{EB}}
              \left( \sqrt{\rho^E}^{-1} \xi \sqrt{\rho^E}^{-1} \ox \1^B \right)
           \sqrt{\rho^{EB}},
\end{equation}
and the recovered state from $\rho^{AE}$ is
\[\begin{split}
  \omega^{AEB} &= (\id^A \ox R^{E\rightarrow EB})\rho^{AE} \\
           &= \sqrt{\rho^{EB}}
                \left( \sqrt{\rho^E}^{-1} \rho^{AE} \sqrt{\rho^E}^{-1} \ox \1^B \right)
              \sqrt{\rho^{EB}}.
\end{split}\]
This map was used to elucidate the structure of $\rho^{AEB}$~\cite{SSA-eq}:
The result is that there has to exist a decomposition
$E = \bigoplus_j e_j^L \ox e_j^R$ of $E$ as a direct sum of tensor products,
such that
\[
  \rho^{AEB} = \bigoplus_j p_j \sigma_j^{A e_j^L} \ox \tau_j^{e_j^R B}.
\]
(In particular, $\rho^{AB}$ is separable.) Such states have been
called ``quantum Markov chains''~\cite{Accardi}.

The recovery map of Fawzi and Renner~\cite{FawziRenner} looks very
similar to the form (\ref{eq:Petz-I}):
\begin{equation}
  \label{eq:FR-I}
  \widetilde{R}(\xi)
        = V\!\sqrt{\rho^{EB}}\!
              \left( \sqrt{\rho^E}^{-1} U\xi U^\dagger \sqrt{\rho^E}^{-1} \!\ox\! \1^B \right)
           \!\sqrt{\rho^{EB}} V^\dagger,
\end{equation}
with certain unitaries $U$ (on $E$) and $V$ (on $EB$).

The near-equality case of Petz's theorem seems to have attracted only 
little attention until recently, for instance as shown here in the context
of squashed entanglement, or in the approach of Brand\~{a}o and Harrow
to finite quantum de Finetti theorems~\cite{BF:de-Finetti}, or
potentially in considerations of many-body physics~\cite{Kim-IC}.
One notable exception is the case of a pure state $\rho^{ABE}$, for
which $I(A:B|E) = I(A:BE)-I(A:E) \approx 0$ corresponds to 
the treatment of approximate quantum error correction due to
Schumacher and Westmoreland~\cite{SchumacherWestmoreland}.

The conjecture that the Petz recovery map $R$ in Eq.~(\ref{eq:Petz-I})
might yield $\omega^{ABE} \approx \rho^{ABE}$ in trace norm seems to have
been formulated first by Kim~\cite{Kim-personal}, cf.~\cite{LiWinter-historic}:
\begin{equation}
  \label{eq:Kim-I}
  I(A:B|E) \stackrel{?!}{\geq} \Omega\left(\|\rho^{AEB}-(\id\ox R)\rho^{AE}\|_1^2\right).
\end{equation}
See also Zhang~\cite{Zhang} (cf.~\cite{LiWinter-historic} once more) 
for this, who suggested the generalized version
\begin{equation}
  \label{eq:Zhang-I}
  D(\rho\|\sigma)-D(T\rho\|T\sigma) \stackrel{?!}{\geq} \Omega\left(\|\rho-RT\rho\|_1^2\right).
\end{equation}

Berta \emph{et al.}~\cite{BertaSeshadreesanWilde} then proposed the
more natural conjecture with the fidelity on the right hand side of 
Eq.~(\ref{eq:Kim-I}), motivated by the observation that the latter 
is a R\'{e}nyi relative entropy:
\begin{equation}
  \label{eq:BSW-I}
  I(A:B|E) \stackrel{?!}{\geq} -\log F\bigl( \rho^{AEB},(\id\ox R)\rho^{AE} \bigr)^2.
\end{equation}
By the well-known relations connecting fidelity
and trace norm, this would imply Kim's conjecture (\ref{eq:Kim-I}). 
While all of the above conjectures remain open (though supported
by increasing numerical evidence), Fawzi and Renner's Eq.~(\ref{eq:FR}) proves a
variant of the last inequality, with $\widetilde{R}$ instead of $R$.
The crucial point of course is that this new map still only acts
on $E$, and as the identity on $A$.

Similarly, Seshadreesan \emph{et al.}~\cite[Conj.~26 \& Sect.~6.1]{SBW}
suggested the following most general form extending (\ref{eq:Zhang-I}),
encompassing all of the above:
\begin{equation}
  \label{eq:SBW}
  D(\rho\|\sigma)-D(T\rho\|T\sigma) \stackrel{?!}{\geq} -\log F(\rho,RT\rho)^2,
\end{equation}
again motivated by a way of writing both sides of the above
as (R\'{e}nyi) relative entropies or variants thereof.

Since the first arXiv posting of the present paper, statements of this form 
have been proven for slight variants of the Petz recovery map, specifically 
the ``swivelled'' (or ``rotated'') Petz maps (cf.~\cite{DupuisWilde})
\[
  R_t(\xi) = \sigma^{-it} R\left( T(\sigma)^{it} \xi T(\sigma)^{-it} \right) \sigma^{it},
\]
which reduces to the Petz recovery map $R=R_0$ for $t=0$, and their convex
combinations. 
Namely, Wilde~\cite{Wilde-recovery}, invoking the Hadamard three-line theorem, 
shows that there exists a $t\in\RR$ (generally depending on all of $T$, 
$\sigma$ and $\rho$) such that eq.~(\ref{eq:SBW}) 
[and similarly eq.~(\ref{eq:BSW-I})] holds with $R_t$ in place of $R$. 
\[
  D(\rho\|\sigma)-D(T\rho\|T\sigma) \geq \inf_t \bigl(-\log F(\rho,R_tT\rho)^2\bigr),
\]
This was then extended to infinite dimension by Junge 
\emph{et al.}~\cite{universal-map}, and improved to a universal 
average over $t$ rather than the minimum on the right hand side:
\[
  D(\rho\|\sigma)-D(T\rho\|T\sigma) \geq \int {\rm d}t \beta_0(t) \bigl(-\log F(\rho,R_{t/2}T\rho)^2\bigr),
\]
with the probability density $\beta_0(t) = \frac{\pi}{2}(1+\cosh(\pi t))^{-1}$. 

Sutter \emph{et al.}~\cite{SutterTomamichelHarrow} presented an essentially
elementary, yet highly nontrivial, argument proving a lower bound for some
unknown convex combination $\widetilde{R}$ of the $R_t$, and in terms of the
\emph{measured relative entropy}: 
\[
  D(\rho\|\sigma)-D(T\rho\|T\sigma) \geq D_{\mathbb{M}}(\rho\|\widetilde{R}T\rho).
\]
This was again improved by Sutter \emph{et al.}~\cite{SutterBertaTomamichel} 
using complex interpolation tools, yielding
\[
  D(\rho\|\sigma)-D(T\rho\|T\sigma) 
        \geq D_{\mathbb{M}}\left( \rho\,\Big\|\int {\rm d}t \beta_0(t) R_{t/2}T\rho \right),
\]
with the same function $\beta_0$ as above.

\bigskip\noindent
{\bf The classical case.}---%
It is well-known that for classical random variables $XYZ$, conditional
independence, i.e.~$I(X:Z|Y) = 0$, implies that $X$ -- $Y$ -- $Z$
is a Markov chain in that order. Furthermore, this is a robust characterization, 
as the following two inequalities show, which we are going to prove.
They provide much of the motivation for the conjectures and results 
presented in the previous section.

\begin{lemma}
  \label{lem:classical-I}
  If $I(X:Z|Y) = \epsilon$ for a distribution $P(XYZ)$, then there exists a
  Markov chain of the same alphabets,
  with distribution $Q(XYZ) = P(XY)P(Z|Y)$, such that the relative entropy
  distance between $P$ and $Q$ is small:
  $D(P_{XYZ}\| Q) = \epsilon$.
  By Pinsker's inequality, this implies
  $\| P_{XYZ}-Q \|_1 \leq \sqrt{2\ln 2}\sqrt{\epsilon}$.
\end{lemma}

\medskip\noindent
This is a special case of the following more general result:

\begin{metalemma}
  \label{meta-lem:classical}
  For any two probability distributions $P$ and $Q$ on the same set
  $\mathcal{X}$, and a stochastic map $T:\mathcal{X}\rightarrow\mathcal{U}$,
  there exists another stochastic map $R$, called the
  \emph{transpose channel}, and which depends only on $Q$ and $T$,
  such that $RTQ=Q$ and
  \begin{equation}
    \label{eq:glorious-relation}
    D(P\|Q) - D(TP\|TQ) \geq D(P\|RTP).
  \end{equation}
  Furthermore, this is an identity if $T$ is deterministic.

  The transpose channel is defined by the property that
  \[
    T(u|x)Q(x) = R(x|u)\,(TQ)(u), 
  \]
  and this is the classical case of Petz's recovery map.
\end{metalemma}
\begin{proof}
Like many classical entropy inequalities, it is an instance of log-concavity.

We have two probability vectors $P = (p_x)_{x\in\mathcal{X}}$ and
$Q = (q_x)_{x\in\mathcal{X}}$, and a stochastic matrix
$T = [t_{ux}]_{u\in\mathcal{U}}^{x\in\mathcal{X}}$ 
(meaning that for all $x$, $\sum_{u\in\mathcal{U}} t_{ux} = 1$). 
The adjoint of cptp map translates into the linear map given by the 
transpose matrix $T^t$.
Then,
\[
  TP = \left( \sum_{x\in\mathcal{X}} t_{ux}p_x \right)_{u\in\mathcal{U}},
  \quad
  TQ = \left( \sum_{x\in\mathcal{X}} t_{ux}q_x \right)_{u\in\mathcal{U}},
\]
and
\[\begin{split}
  RTP &= \Bigl( q_x \Bigl(T^t \bigl( (TP)_u/(TQ)_u \bigr)_{u\in\mathcal{U}}\Bigr)_x \Bigr)_{x\in\mathcal{X}} \\
      &= \left( q_x \sum_{u\in\mathcal{U}} t_{ux}\frac{\sum_{x'} t_{ux'}p_{x'}}
                                                      {\sum_{x'} t_{ux'}q_{x'}} \right)_{x\in\mathcal{X}},
\end{split}\]
leading to the following expressions for the three relative
entropies concerned:
\begin{align*}
  D(P\|Q)    &= \sum_{x\in\mathcal{X}} p_x \log\frac{p_x}{q_x}, \\
  D(TP\|TQ)  &= \sum_{u\in\mathcal{U}} \left(\sum_{x\in\mathcal{X}} t_{ux}p_x\right)
                                      \log\frac{\sum_{x'} t_{ux'}p_{x'}}{\sum_{x'} t_{ux'}q_{x'}}, \\
  D(P\|RTP)  &= \sum_{x\in\mathcal{X}} p_x \log\left( \frac{p_x}{q_x}
                              \frac{1}{\sum_u t_{ux}\frac{\sum_{x'} t_{ux'}p_{x'}}{\sum_{x'} t_{ux'}q_{x'}}} \right).
\end{align*}

The claimed inequality, that the first expression is larger or equal
to the sum of the last two, can be rearranged as
$D(P\|Q) - D(P\|RTP) \geq D(TP\|TQ)$,
which simplifies to
\[
  \sum_{x\in\mathcal{X}} p_x \log\left( \sum_{u\in\mathcal{U}} t_{ux}
                                           \frac{\sum_{x'} t_{ux'}p_{x'}}{\sum_{x'} t_{ux'}q_{x'}} \right) \\
                  \geq \sum_{x\in\mathcal{X}} p_x \sum_{u\in\mathcal{U}} 
                                               t_{ux} \log\frac{\sum_{x'} t_{ux'}p_{x'}}{\sum_{x'} t_{ux'}q_{x'}}.
\]
However, this is true for each term $x$, due to the concavity of
the logarithm, and $\sum_u t_{ux} = 1$.

It can be checked from this that if the channel $T$ is deterministic,
i.e.~if for each $x\in\mathcal{X}$ there is only one $u\in\mathcal{U}$ such that $t_{ux} > 0$,
then equality holds; in particular this is the case where $T$ is the marginal
map from $\mathcal{X}\times \mathcal{Y}$ to $\mathcal{X}$.
\qed
\end{proof}

\medskip
Observe that the inequality (\ref{eq:glorious-relation})
implies the conjectures (\ref{eq:Kim-I}), (\ref{eq:Zhang-I}), 
(\ref{eq:BSW-I}) and (\ref{eq:SBW}) in the classical case, because 
of $D(P\|Q) \geq -\log F(P,Q)^2$.
The results of~\cite{Wilde-recovery} and \cite{SutterTomamichelHarrow}
reproduce this relaxed version of the classical case, 
because when restricted to diagonal
density matrices, the swivelled Petz maps $R_t$ reduce to $R_0=R$
for all $t$.
Notably the approach of~\cite{SutterTomamichelHarrow} is strikingly
close to our above classical proof by log-concavity, using pinching to
remove non-commutativity and otherwise using only operator
monotonicity and concavity of the logarithm; at the same time it relies on
looking at asymptotically many copies of the state, which is one of the
reasons why $-\log F$ appears in the end result rather than the
relative entropy.

It is known, by numerical counterexamples, that (\ref{eq:glorious-relation})
is false in the quantum case, already for qubits, and also restricting
to the case $T=\tr_B$, $\rho = \rho^{AEB}$ and
$\sigma = \rho^A \ox \rho^{EB}$~\cite{Kim-personal}.
However, one might be tempted to speculate that with a variant of the Fawzi-Renner
map, say some $\widehat{R}$ (perhaps even a rotated or averaged Petz map $R_t$), 
we might have
\begin{equation}
  \label{eq:KeWinter-I}
  I(A:B|E) \stackrel{?!}{\geq} D\bigl( \rho^{AEB} \| (\id\ox\widehat{R})\rho^{AE} \bigr),
\end{equation}
which would also imply (\ref{eq:FR}). 
However, since the first circulation of our earliest unpublished 
notes~\cite{LiWinter-historic}, 
this conjecture has been subjected to serious scrutiny, and recently
Fawzi and Fawzi~\cite{Fawzi2-CMI} have found an explicit counterexample by 
rigorous numerical computer calculations: there does not exists
a map $\widehat{R}$ recovering $\sigma$, i.e.~$\widehat{R}T\sigma=\sigma$
and at the same time satisfying Eq.~(\ref{eq:KeWinter-I}).

\bigskip\noindent
{\bf Discussion.}---%
We have shown how Fawzi and Renner's recent breakthrough in the characterization
of small quantum conditional mutual information has consequences for
the faithfulness of squashed entanglement. We believe that the same
approach can be used also to address the faithfulness of the multi-party
squashed entanglement~\cite{Yang-et-al}, however technical issues remain,
which are explained in the Appendix.

The breakthrough of~\cite{FawziRenner}, and the subsequent results, 
also finally clarify the ``right'' robust version of quantum 
Markov chains, which are equivalently 
characterized by $I(A:B|E)\approx 0$ and by the existence of a recovery
map such that $\rho^{AEB} \approx (\id_A\ox\widetilde{R})\rho^{AE}$,
cf.~\cite[Prop.~35]{BertaSeshadreesanWilde}.
For classical probability distributions,
yet another way of expressing this is to say that there exists a
Markov chain close to the given density, but this is not the case
in the quantum analogue~\cite{IbinsonLindenWinter,antisymm}, at
least if one wants to avoid introducing strong dimensional dependence.

\medskip
To conclude, looking back at the conjectures and theorems 
reviewed above, and contrasting them with the clear picture emerging from
the classical case, we wish to suggest a target for further investigation,
which takes us in a direction different from the conjecture
(\ref{eq:SBW}) and its descendants.

Namely, the question is, whether it is possible to
define a recovery map $\widehat{R}=\widehat{R}(T,\sigma)$ for
every pair of a cptp map $T$ and a state $\sigma$ in its domain,
such that $\widehat{R}T\sigma=\sigma$ and
\begin{equation}
  \label{eq:big-one}
  D(\rho\|\sigma)-D(T\rho\|T\sigma) \stackrel{?!}{\geq} \widetilde{D}\bigl(\rho\|\widehat{R}T\rho\bigr),
\end{equation}
with a suitable divergence $\widetilde{D}$,
and such that the following functoriality properties hold.
\begin{itemize}
  \item \emph{Normalization}:
    To the identity map $\id$ and any state (of full rank), the identity map
    is associated: $\widehat{R}(\id,\tau) = \id$.
  \item \emph{Tensor}:
    If $\widehat{R}_i = \widehat{R}(T_i,\sigma_i)$ is
    associated to maps $T_i$ and states $\sigma_i$,
    then the map associated to $T_1\ox T_2$ and state $\sigma_1\ox\sigma_2$,
    is $\widehat{R}(T_1\ox T_2,\sigma_1\ox\sigma_2) = \widehat{R}_1 \ox \widehat{R}_2$.
\end{itemize}

This would clearly imply the inequality (\ref{eq:KeWinter-I}), with
$\widetilde{D}$ in place of $D$. Hence, it cannot be true for the usual
(Umegaki) relative entropy \cite{Fawzi2-CMI}.
Note that the Petz map quite evidently obeys the functoriality properties, in fact
in addition also another one:

\begin{itemize}
  \item \emph{Composition}:
    For cptp maps $T_i$ on suitable space, such that we can form
    their composition $T_2 \circ T_1$, and a state $\sigma$ such that
    we have associated maps $\widehat{R}_1 = \widehat{R}(T_1,\sigma)$
    and $\widehat{R}_2 = \widehat{R}(T_2,T_1\sigma)$,
    we have
    $\widehat{R}(T_2\circ T_1,\sigma) = \widehat{R}_1 \circ \widehat{R}_2$.
\end{itemize}

Can all these constraints be satisfied simultaneously? And if so, what
would be the applications of such a result? Note that the Petz
recovery map is a very useful tool in ``pretty good'' state discrimination 
and quantum error correction~\cite{BarnumKnill,SchumacherWestmoreland};
the functoriality above along with (\ref{eq:big-one}) is meant to
preserve these good properties. The current status of this question is 
the following: We know that one can indeed define a ``universal'' recovery 
map $\widehat{R}$ for inequality (\ref{eq:big-one}), 
with either $\widetilde{D}=-\log F$ or $\widetilde{D}=D_{\mathbb{M}}$
-- in fact in the convex hull of the swivelled Petz maps $R_t$ --, where 
universality refers to the map depending only on $T$ and $\sigma$. 
It furthermore satisfies the normalization property, as well as 
tensorization with the identity~\cite{SutterBertaTomamichel,universal-map}.

\bigskip\noindent
\begin{acknowledgements}
Since the first formulation of the proof idea of Theorem \ref{thm:main} 
in 2008~\cite{LiWinter-historic},
we have enjoyed conversations and the keen interest of many people in
the question of recoverability and remainder terms for strong subadditivity
and the monotonicity of relative entropy,
among them Fernando Brand\~{a}o, Matthias Christandl, Paul Erker, Omar Fawzi, 
Aram Harrow, Isaac Kim, C\'ecilia Lancien, Renato Renner and Mark Wilde.

When this work was started, KL was affiliated with the Centre for Quantum Technologies
(CQT), National University of Singapore; AW was affiliated with the School
of Mathematics, University of Bristol and with CQT.
KL was supported by NSF grants CCF-1110941 and CCF-1111382.
AW was supported by the ERC Advanced Grant ``IRQUAT'', the European
Commission (STREP ``RAQUEL''), the Spanish MINECO (grants 
FIS2008-01236, FIS2013-40627-P and FIS2016-86681-P) 
with the support of FEDER funds,
and the Generalitat de Catalunya CIRIT, project 2014-SGR-966.
\end{acknowledgements}

\section*{Appendix --- Multi-party squashed entanglement.}
One might wonder if our approach could also be used to prove faithfulness 
of the multi-party squashed entanglement~\cite{Yang-et-al},
\begin{equation}
  E_{\text{sq}}(\rho^{A_1\ldots A_n}) = \inf_{\rho^{A_1\ldots A_n E}} \frac{1}{2}I(A_1:\cdots:A_n|E),
\end{equation}
with 
$I(A_1\!:\!\cdots\!:\!A_n|E) = \sum_{i=1}^n S(A_i|E) - S(A_1\ldots A_n|E)$
the conditional multi-information.
That is, $E_{\text{sq}}(\rho^{A_1\ldots A_n})$ would vanish iff 
$\rho$ is $n$-separable:
\[
  \rho^{A_1\ldots A_n} = \sum_\lambda p_\lambda \rho_{\lambda|1}^{A_1} \ox\cdots\ox \rho_{\lambda|n}^{A_n}.
\]
It seems that with the methods of~\cite{BCY,LiWinter} this cannot be approached.

The idea starts from the identity
\[\begin{split}
  I(A_1:\cdots:A_n|E)
        &= I(A_1:A_2\ldots A_n|E) + I(A_2:\cdots:A_n|E) \\
        &= \ldots = \sum_{i=1}^{n-1} I(A_i:A_{i+1}\ldots A_n|E),
\end{split}\]
showing that $I(A_1:\cdots:A_n|E) \leq 2\epsilon$ implies,
for all $i$, $I(A_i:A_{[n]\setminus i}|E) \leq 2\epsilon$, and more
generally, for all subsets $I\subset [n]$, $I(A_I:A_{[n]\setminus I}|E) \leq 2\epsilon$.

In particular, if $\epsilon = 0$, we can use the structure theorem of
\cite{SSA-eq} to find, for each $i$, a projective measurement 
$\bigl( P^{(i)}_{\lambda_i} \bigr)$ on $E$ that commutes with $\rho^{A_1\ldots A_n E}$,
such that for all $\lambda_i$,
\[
  \tr_E \rho^{A_1\ldots A_n E} P^{(i)}_{\lambda_i} 
      = p_{\lambda_i} \sigma_{\lambda_i}^{A_i} \ox \tau_{\lambda_i}^{A_{[n]\setminus i}}, 
\]
i.e., conditioned on the measurement outcomes $\lambda_i$, $A_i$ and $A_{[n]\setminus i}$
are in a product state.
Performing all these measurements in some fixed order, we thus obtain outcomes
$\lambda = \lambda_1\ldots\lambda_n$ such that conditioned on $\lambda$, the
state is a product state with respect to all partitions $i:[n]\setminus i$, 
which means that conditioned on $\lambda$, $A_1,\ldots,A_n$ factorize.

\medskip
We would like to use the machinery of the recovery maps to 
extract from $E$ a large number $k$ of
approximate copies of each $A_i$, using approximate recovery maps
$\widetilde{R}_i:\mathcal{L}(E) \longrightarrow \mathcal{L}(EA_i)$
according to Eq.~(\ref{eq:FR}).
With $t = \sqrt{8\ln 2}\sqrt{\epsilon}$ and tracing out
$E$, we can indeed get a state $\Omega^{A_1 A_2^{[k]}\ldots A_n^{[k]}}$,
with $A_i^{[k]} = A_i^1\ldots A_i^k$ consisting of $k$ copies
of $A_i$, such that
\[
  \| \rho^{A_1\ldots A_n} - \Omega^{A_1 A_2^{j_2}\ldots A_n^{j_n}} \|_1
                     \leq (n-1)(k-1)t \leq nk\sqrt{8\ln 2}\sqrt{\epsilon},
\]
for all tuples $(j_2,\ldots,j_n)$ such that all but at most one $j_i$ equals $1$.

We cannot say easily that this holds for all tuples $(j_2,\ldots,j_n)$,
because the different recover maps may interfere with each other.
However, if we could conclude that, we would be done: by symmetrizing
the $k$ copies of each $A_i$ ($i>1$) we would find, as before,
that $\rho$ is $O(\sqrt[4]{\epsilon})$-close to a $k$-extendible 
state, with $k=\Omega\left(\frac{1}{\sqrt[4]{\epsilon}}\right)$.

We could then again use the results of~\cite{NOP}, now extended to the multi-partite
case, to see that $\Omega^{A_1 A_2^{j_2}\ldots A_n^{j_n}}$ is at trace distance
at most $\frac{2}{k}(|A_2|^2+\cdots+|A_n|^2)$ from a fully separable (i.e.~$n$-separable)
state. Note that a reasoning along these lines goes through for the
-- generally larger -- multi-party conditional entanglement of mutual information 
(CEMI)~\cite{YangHorodeckiWang,Yang-et-al} 
\[
  E_{\text{I}}(\rho^{A_1\ldots A_n}) 
              = \inf_{\rho^{A_1 A_1'\ldots A_n A_n'}} 
                  \frac{1}{2}\bigl[ I(A_1A_1':\cdots:A_n A_n') - I(A_1':\cdots:A_n') \bigr],
\]
as since
shown by Wilde~\cite{Wilde-CEMI}. We have to leave the problem of
finding an extension of Theorem~\ref{thm:main} to $n>2$ parties to the 
attention of the interested reader.

\end{document}